\documentclass[epj]{svjour}
\usepackage{graphics}
\usepackage{times}
\usepackage{amsmath}
\usepackage{amssymb}
\usepackage{mathrsfs}
\usepackage{graphicx}
\usepackage[justification=centering]{caption}  
\usepackage{float}  
\usepackage{cite}
\usepackage{upref}
\makeatletter
\let\c@lofdepth\relax
\let\c@lotdepth\relax
\makeatother
\usepackage{subfigure}
\usepackage[titles,subfigure]{tocloft}
\begin{document}
\title{Dynamics analysis of a novel fractional HBV system with immune delay}
\subtitle{}
\author{Xiaohe Yu\inst{1} \and Fei Gao\inst{1,2,}%
\thanks{{e-mail:} hgaofei@gmail.com}%
\and Xingluan Wang\inst{1}\and Jiangtao Zhang\inst{1}
}                     
\institute{School of Science, Wuhan University of Technology, Wuhan 430070, China \and Center for Mathematical Sciences, Wuhan University of Technology}
\date{Received: date / Revised version: date}
%
\abstract{
In the theoretical research of Hepatitis B virus, mathematical models of its transmission mechanism have been thoroughly investigated, while the dynamics of the immune process in vivo have not. At present, nearly all the existing models are based on integer-order differential equations. However, models restricted to integer-order cannot depict the complex dynamical behaviors of the virus. Since fractional-order models possess property of memory and memory is the main feature of immune response, we propose a fractional-order model of Hepatitis B with time-delay to further explore the dynamical characters of the HBV model. First, by using the Caputo fractional differential and its properties, we obtain the existence and uniqueness of the solution. Then by utilizing stability analysis of fractional-order system, the stability results of the disease-free equilibrium and the epidemic equilibrium are studied according to the value of the basic reproduction number. Moreover, the bifurcation diagram, the largest Lyapunov exponent diagram, phase diagram, Poincar\'e section and frequency spectrum are employed to confirm the chaotic characteristics of the fractional HBV model and to figure out the effects of time-delay and fractional order. Further, a theory of the asymptotical stability of nonlinear autonomous system with delay on the basis of Caputo fractional differential is proved. The results of our work illustrates the rich dynamics of fractional HBV model with time-delay, and can provide theoretical guidance to virus dynamics study and clinical practice to some extent.
\PACS{
      {05.45.-a}{Nonlinear dynamics and chaos}   \and
      {05.45.Pq}{Numerical simulation of chaotic systems}
     } 
} 
%
\maketitle
\section{Introduction}
Hepatitis B is a major disease worldwide which is highly contagious \cite{TSAI201833}. The disease mainly threatens children and young adults, and a small number of patients develop chronic infection which can convert to cirrhosis or liver cancer. According to the information published by WHO on World Hepatitis Day 2019, Hepatitis B is the second most deadly epidemic after tuberculosis, and the number of people infected with hepatitis virus is 9 times higher than that of HIV \cite{THORNLEY2008599}. However, as reported by the EASL International Liver Congress$^{TM}$ 2019, ILC 2019, hepatitis B treatment coverage is low among countries in all income strata. As of 2016, it is estimated that only 10.5\% of all people to be living with hepatitis B were aware of their infection, whereas 16.7\% of the people diagnosed were under treatment. Overall, there are mainly two treatments for chronic hepatitis B patients, i.e. treatment with a NA or with IFN$\alpha$, currently pegylated (PegIFN$\alpha$)\cite{2treatments2019,Friedman2018}. Despite of the prevalence of the use of oral medications such as tenofovir or entecavir, in most people, the treatment does not cure hepatitis B infection, but only suppresses the replication of the virus.

The prevention and treatment of hepatitis B is a great issue related to human health and livelihood. Quantitative research on the transmission rules of hepatitis B is an important basis for its prevention and treatment. Kinetics of Hepatitis B virus (HBV) is an important method for theoretical quantitative research on the prevalence of hepatitis B. Compared with traditional biostatistical methods, infectious disease dynamics methods can better reflect the mechanism of disease diffusion while enabling people to understand some global states in the epidemic process \cite{QiaoMeiHong2010}. Mathematical models can be used to make predictions, and various virtual experiments can also be performed. In recent years, the analysis of mathematical models of epidemics has received extensive attention from scholars. Through the establishment of infectious disease dynamics system researchers analyze the development of epidemics, thus determining corresponding prevention and control strategies from the obtained rules. Numerical simulation methods are applied to verify the validity of the models. The existing related literature has done abundant researches on mathematical models of HBV, especially on integer-order systems. For instance, Khan \emph{et al.} \cite{Khan2017} formulated a transmission dynamic model dividing the host population into four compartments which includes individuals infected with acute hepatitis and those who infected with chronic. Kamyad \emph{et al.} \cite{Kamyad2014} considered a SEICR epidemic model, obtaining the optimal control strategy to minimize both the number of infectious humans and the vaccination and treatment costs. In \cite{Chenar}, Chenar \emph{et al.} took the innate and adaptive immune responses into account whilst analyzing stability of various steady states. Liu Peng \cite{LiuPeng2013WangWenDi} established a model to explore the influence of age structure and vaccination of HBV infection. Chaotic behaviors of an integer-order HBV model are demonstrated in \cite{WangKaiFaMain} through theoretical and numerical studies. Investigation of the stability of equilibria and persistence of the contagion by the basic reproductive number $\mathscr{R}_{0}$ can be found in \cite{JingZhangOptimalCVT}. The work mentioned above reveal HBV's pathogenesis and epidemic regularity to some extent, but all of them are restricted to integer-order systems. 

Fractional-order system not only extends the traditional integer-order system theory, but also is capable of better describing many real-world processes. Scholars have pointed out that systems can be characterized more accurately by fractional-order differential equations. The utilization of fractional calculus in physics and engineering has been greatly developed. In paper \cite{ReviewST}, the authors generalized the recent trends in theory and applications of fractional dynamics.  Presenting a new fractional-order inductive transducer, Veeraian \emph{et al.} \cite{Veeraian} expounded the impact of fractional-order parameters on the performance of inductive transducer. Article \cite{RAHIMKHANI20168087} defined a new fractional-order Bernoulli wavelets then expanded the application of the proposed wavelets for solving fractional differential systems. Provided that fractional-order models possess virtues of memory, non-local and non-singular kernel, more recently they are gradually utilized to mimic biological processes which contain memory phenomena such as epidemiological models. In article \cite{Angstmann2016}, Angstmann \emph{et al.} obtained a fractional-order recovery SIR model based on the derivation from a stochastic process. 

By introducing a fractional order into MINMOD Millennium. Cho \emph{et al.} \cite{CHO201536} interpreted the rheological behavior of glucose and insulin to estimate insulin sensitivity. From the view of signal processing techniques, Machado \cite{MACHADO20154095} applying fractional calculus to describe the essential characteristics of DNA. Employing the Gr{\"u}nwald-Letnikov fractional differential definition, Xu \cite{QPSO} introduced a new approach to enhance the global search ability of quantum-behaved particle swarm optimization(QPSO). In paper \cite{GaoFei2016FWAS}, the authors studied a fractional Willis aneurysm system and used two methods to control the chaotic system. Rosa and Torres \cite{ROSA2018142} established a nonlocal fractional-order model of an implemented human respiratory syncytial virus surveillance system and solved an optimal treatment control problem. By applying fractional differential, the cited work mimicked biological processes and confirmed to achieve valuable results. 

As a matter of fact, fractional HBV models have not yet been studied extensively. To interpret the spread of HBV, Shi \emph{et al.} \cite{Owo} considered a model from the perspective of fractional-in-space reaction{\textendash}diffusion equations. Substituting the bilinear incidence rate with Holling type-II functional response, paper \cite{Holling} analyzed the control of the disease dissemination. A research on a HBV model using the Caputo-Fabrizio derivative was carried out by Ullah \cite{Ullah2018CF}. Later the same model employing the Atangana-Baleanu derivative was investigated in \cite{Ullah2018AB}. Motivated by the work above, in the present paper, we propose a fractional-order HBV model with immune delay using Caputo fractional differential.

When working with time related phenomena, time-delay is an inevitable consideration. In engineering practice, although the amount of time-delay is small, it often affects the stability of the entire system. Therefore, under this circumstances, the effect of time-delay on dynamics system is an unavoidable factor. Gao first researched the fractional Willis aneurysm system in \cite{GaoFei2016FWAS}, then compared the aforementioned system with the one including time-delay factors \cite{GaoFei2018delay}, and validated the effectiveness of the time-delay system. When exploring the impulsive fractional differential systems, Nieto \cite{priceFluctuations} introduced delay to model price fluctuations, while Stamov \cite{Stamov2017early} investigated a fractional neural networks with time-varying delays. The basic feature of a time-delay dynamic system is that the evolution of the system over time depends not only on the current state of the system, but also on its past state. During the invasion of virus, there are various time delays in the process of infection of susceptible cells and reactions of immune defense. Hence, the introduction of time-delay will be more in line with the real situations of viral infection and immune response, which is more conducive to understanding the course of infection thus achieving the purpose of prevention and treatment. Compared with system without time-delay, the infinite dimensional dynamic system described by delay differential equations has crucial differences in many aspects. Showing complicated phenomena, the dynamic behavior of fractional delay systems is not completely the same with the one without delay. Simple time-delay differential systems can have rather complex dynamical properties. Time-delay often leads to instability of the system and generates various forms of bifurcations. To observe the bifurcation and chaos phenomena, scholars employ numerical experiments by drawing phase portraits of the system \cite{COMPLEXITY,Abdelaziz2018,LocalRosslerwithnodelay}, bifurcation plots with certain parameters \cite{JI2018352,Abdelaziz2018,Ding2017EPJfracbifchao}, curves of time course \cite{LocalRosslerwithnodelay,Gao_2019AB}, largest Lyapunov exponent diagram \cite{GaoFei2018delay,Abdelaziz2018}, power spectrum \cite{Ding2017EPJfracbifchao} etc. to investigate fractional-order models. Nevertheless, none of the existing fractional-order HBV model has interpreted its instable state from the point of bifurcation and chaos or expounded its chaotic characteristics. As a result, in this paper, we introduce an essential delay in the course of cellular immunity and utilize dynamical analysis techniques to study its impact on the system stability, and further discuss the chaotic properties of the proposed system.

Based on one of the most frequently used fractional differential definitions, Caputo differential, we propose a fractional HBV system with immune delay. In the present paper, we first introduce some definitions and properties of the Caputo derivative. In sect. 2, we select a typical also well-studied integer-order HBV model then fractionalize the order of the system. Next we show the existence and uniqueness of the system solution, and analyze the stability of the equilibria in sect. 3. Sect. 4 demonstrates numerical simulation, investigates the effect of time-delay and fractional-order on the HBV model. Moreover, we illustrates the complex dynamical behavior and the chaotic characteristics of the Caputo fractional-order system, confirming that it possesses the similar chaotic properties as the integer-order system does. In sect. 5, regarding to the derived model, we prove an asymptotical stability theory of Caputo fractional system with fixed time delay. Finally, we conclude our work in sect. 6.
\section{Preliminaries of the Caputo fractional-order derivative}
\begin{definition}
Let $\alpha$  be a positive real number, $n-1\leq \alpha < n $ , $n \in \mathrm{N}_{+}$. The function $f$ is defined on the interval $[a, b]$. Then the Caputo fractional derivative of order $\alpha$ is defined by
\begin{equation*}
\label{equ1}
  _{t_{0}}^{c} D_{t}^{\alpha} f(t)=\frac{1}{\Gamma(n-\alpha)} \int_{t_{0}}^{t} \frac{f^{(n)}(\tau)}{(t-\tau)^{\alpha-n+1}} \mathrm{d} \tau, t \in[a, b],
\end{equation*}
where $\alpha$  is the order of the derivative and $\Gamma(z)$ denotes the Gamma function.
\end{definition}
\begin{theorem}\label{TH1}
	\cite{HuoRan2014TH} Given system with time-delay $\tau$  and initial values
\begin{equation*}
\label{equ2}
\left\{ \begin{array}{l}
\vspace{5pt}
{}_{{t_0}}^CD_t^\alpha x(t) = f(t,x(t),x(t - \tau )),t \ge {t_0},\\
{x^{(k)}}(t) = {\varphi _k}(t),{t_0} - \tau  \le t \le {t_0},
\end{array} \right.
\end{equation*}
if $f$  is continuous in a neighborhood of $\left(t_{0}, \varphi\left(t_{0}\right), x\left(t_{0}-\tau\right)\right)$  and is Lipschitz continuous for all variables expect $t$, $\varphi_{k}(t) \in C\left[t_{0}-\tau, t_{0}\right]$  , then the system has a unique continuous solution on $t_{0} \leq t \leq t_{0}+h$ ,where h is sufficiently small.
\end{theorem}
\begin{lemma}\label{lemma1}
	\cite{AHMED2007542} $E$ is locally asymptotic stable if all the eigenvalues $\lambda$ of the Jacobian matrix $J(E)$ at the equilibrium point satisfy $|\arg (\lambda)|>\dfrac{\alpha \pi}{2}$.\\
\end{lemma}
\begin{lemma} \label{lemma2}
	\cite{DUARTEMERMOUD2015650} Let $x(t) \in R^{n}$ be a vector of differentiable functions. Then, for any time instant $t \geq t_{0}$ , the following relationship holds
	\begin{equation*}
	\dfrac{1}{2}{}_{{t_0}}^{}D_t^\alpha ({x^T}(t)Px(t)) \le {x^T}(t)P{}_{{t_0}}^{}D_t^\alpha x(t), \forall \alpha \in(0,1], \forall t \geq t_{0}
	\end{equation*}   
	where $P \in R^{n \times n}$ is a constant, square, symmetric and positive definite matrix.
\end{lemma}
\section{Fractional HBV model with immune delay}
We first introduce the classical integer-order model of HBV \cite{WangKaiFaMain}
\vspace{5pt}
\begin{equation}\label{E1}
  \left\{\begin{array}{l}
  \vspace{5pt}
  {x^{\prime}(t)=\lambda-d x(t)-\beta x(t) y(t)}, \\
  \vspace{5pt}
  {y^{\prime}(t)=\beta x(t) y(t)-a y(t)-q y(t) z(t)}, \\
  {z^{\prime}(t)=c y(t-\tau)-b z(t)},\end{array}\right.
\end{equation}
where $x(t)$ denotes the number of susceptible cells, $y(t)$ denotes the number of HBV population, and $z(t)$ denotes the number of cytotoxic T lymphocytes (CTL). Susceptible host cells are produced at a rate $\lambda$, die at a rate $dx$ and get infected by the virus at a rate $\beta xy$. Infected cells' natural death rate is $ay$ and are killed by CTL reaction at a rate $pyz$, corresponding to the lysis function of CTL response. The activation rate of the CTL response is proportional to the number of infected cells at a previous time, $cy(t-\tau)$, where $\tau$ is the time delay of CTL response. The response also decays at an exponential rate proportional to its current strength, $bz$. Owing to the evidence that CTLs play the key role in the course of immunization by attacking virus-infected cells, we reserve the original time delay, namely the delay caused by the previous population of the CTLs, then apply Caputo fractional differential operator to propose a novel fractional HBV model
\begin{equation}\label{E2}
\left\{ \begin{array}{l}
\vspace{5pt}
{}_{{t_0}}^CD_t^\alpha x(t) = \lambda  - dx(t) - \beta x(t)y(t),\\
\vspace{5pt}
{}_{{t_0}}^CD_t^\alpha y(t) = \beta x(t)y(t) - ay(t) - qy(t)z(t),\\
\vspace{5pt}
{}_{{t_0}}^CD_t^\alpha z(t) = cy(t - \tau ) - bz(t),\\
0 < \alpha  < 1.
\end{array} \right.
\end{equation}
\subsection{Existence and uniqueness of the solution}
\begin{theorem}\label{TH2}
	The system with initial value condition is constructed as follows:
	\begin{equation*}\label{31}
	\left\{\begin{array}{c}
	\vspace{5pt}
	{_{0} D_{t}^{\alpha} X(t)=F(t)+F_{1}(X(t-\tau))+\mathrm{F}_{2}(X(t))}, \\
	{X(0)=X_{0}}.\end{array}\right.
	\end{equation*}
	where
	\begin{equation*}\label{32}
	X(t)=(x(t), y(t), z(t))^{2},
	\end{equation*}
	\begin{equation*}\label{33}
	X(t-\tau)=(x(t-\tau), y(t-\tau), z(t-\tau))^{T},
	\end{equation*}
	\begin{equation*}\label{34}
	X_{0}=\left(x_{0}, y_{0}, z_{0}\right)^{T},
	\end{equation*}
	\begin{equation*}\label{35}
	F(t)=\left( \begin{array}{l}{\lambda} \\ {0} \\ {0}\end{array}\right),
	\end{equation*}
	\begin{equation*}\label{36}
	F_{1}(X(t-\tau))=\left( {\left( {\begin{array}{*{20}{c}}
			0&0&0\\
			0&0&0\\
			0&c&0
			\end{array}} \right)X(t - \tau )} \right) = \left( \begin{array}{c}
	0\\
	0\\
	cy(t - \tau )
	\end{array} \right),
	\end{equation*}
	\begin{equation*}\label{37}
	{{\rm{F}}_2}{\rm{(X(t)) = }}\left( {\left( {\begin{array}{*{20}{c}}
			{ - d}&{ - \beta x(t)}&0\\
			{\beta y(t)}&{ - a}&{ - qy(t)}\\
			0&0&{ - b}
			\end{array}} \right)X(t)} \right) = \left( \begin{array}{c}
	- dx(t) - \beta x(t)y(t)\\
	\beta x(t)y(t) - ay(t) - qy(t)z(t)\\
	{\rm{ - b z(t)}}
	\end{array} \right),
	\end{equation*}
	then the system with initial value conditions has a unique solution. 
\end{theorem}                                                                                          
\begin{proof}
	Take $|\cdot|$ and $||\cdot||$ as vector norm and matrix norm, respectively. Let $G(t,X(t))=F(t)+F_{1}(X(t-\tau))+F_{2}(X(t))$\quad$. For any \sigma>0,  \quad\left[x_{0}-\sigma, x_{0}+\sigma\right]$ is continuous and bounded. For $\forall X(t), Y(t) \in\left[x_{0}-\sigma, x_{0}+\sigma\right]$, there is
	\begin{equation*}\label{38}
	|G(t, X(t))-G(t, Y(t))| \leq\left|F_{1}(X(t-\tau))-F_{1}(Y(t-\tau))\right|+\left|F_{2}(X(t))-F_{2}(Y(t))\right|. 
	\end{equation*}
	As for $\left|F_{2}(X(t))-F_{2}(Y(t))\right|$,
	there exists
	\begin{align}
	\left| {{F_2}(X(t)) - {F_2}(Y(t))} \right|& = {\left|\left( {\begin{array}{*{20}{c}}
			0&{ - \beta {x_1}(t) + \beta {x_2}(t)}&0\\
			\vspace{5pt}
			{\beta {y_1}(t) - \beta {y_2}(t)}&0&{ - q{y_1}(t) + q{y_2}(t)}\\
			0&0&0
			\end{array}} \right)\right|}\notag\\
	&\le \beta \left| {{y_1}(t) - {y_2}(t)} \right| + \beta \left| {{x_1}(t) - {x_2}(t)} \right| + q\left| {{y_1}(t) - {y_2}(t)} \right|\notag\\
	& \le (\beta  + q)\left| {{y_1}(t) - {y_2}(t)} \right| + (\beta  + q)\left| {{x_1}(t) - {x_2}(t)} \right|\notag\\
	& = (\beta  + q)|X(t) - Y(t)|.
	\end{align}
	Let $L_{1}=\beta+q $. Since $\left[x_{0}-\sigma, x_{0}+\sigma\right]$ is continuous and bounded, $\beta, q$ are parameters within certain ranges, thereby $L_1$ is a constant. 
	Similarly, we have $\left|F_{2}(X(t))-F_{2}(Y(t))\right| \leq L_{1}|X(t)-Y(t)|$ for $| F_{2}(X(t))-F_{2}(Y(t)) |$.\\
	As for $\left|F_{1}(X(t-\tau))-F_{1}(Y(t-\tau))\right|$, there is
	\begin{align}
	\vspace{5pt}
	{F_1}(X(t - \tau )) - {F_1}(Y(t - \tau )) &= \left|\left( {\begin{array}{*{20}{c}}
		\vspace{5pt}
		0\\
		0\\
		{c\left( {{y_1}(t - \tau ) - {y_2}(t - \tau )} \right)}
		\end{array}} \right)\right|\notag\\
	&= c\left| {{y_1}(t - \tau ) - {y_2}(t - \tau )}\right|\notag\\
	&\le c|X(t - \tau ) - Y(t - \tau )|.
	\end{align}
	Let $L_2=c$. Likewise,
	\begin{equation*}\label{hgfsd}
	\left|F_{1}(X(t-\tau))-F_{1}(Y(t-\tau))\right| \leq L_{2}|X(t)-Y(t)|.
	\end{equation*}
	Let $L=\max \left\{L_{1}, L_{2}\right\}$, then
	\begin{equation*}\label{561}
	|G(t, X(t))-G(t, Y(t))| \leq(|X(t)-Y(t)|+|X(t-\tau)-Y(t-\tau)|).
	\end{equation*}
	Thus $G(t, X(t))$ is Lipschitz continuous for all variables except $t$. 
	According to the known conditions, $G(t, X(t))$  is continuous in the neighborhood of the given initial value, hence $G(t, X(t))$  satisfies theorem \ref{TH1}, i.e. the system with the initial value has a unique solution. Theorem \ref{TH2} is verified. 
\end{proof}

\subsection{Stability of the equilibria}
Let $D_{t}^{\alpha} x(t)=0, D_{t}^{\alpha} y(t)=0, D_{t}^{\alpha} z(t)=0$, we obtain three equilibria: The infection-free equilibrium $E_{0}\left(\dfrac{\lambda}{d}, 0,0\right)$,
the infection equilibria
\begin{equation*}\label{441}
  E_{1}\left(\frac{1}{\beta}\left[a-\frac{1}{2 b \beta}\left(T_{3}-T_{1}\right)\right],-\frac{1}{T_{2}}\left(T_{3}-T_{1}\right),-\frac{1}{T_{4}}\left(T_{3}-T_{1}\right)\right),
\end{equation*}
and
\begin{equation*}\label{441}
E_{2}\left(\frac{1}{\beta}\left[a-\frac{1}{2 b \beta}\left(T_{3}-T_{1}\right)\right],-\frac{1}{T_{2}}\left(T_{1}+T_{3}\right),-\frac{1}{T_{4}}\left(T_{1}+T_{3}\right)\right),
\end{equation*}
where
\begin{equation*}\label{fdfd}
 \begin{array}{l}
{T_1} = \sqrt {{a^2}{b^2}{\beta ^2} - 2ab\beta cdq + 4\lambda b{\beta ^2}cq + {c^2}{d^2}{q^2}}  > 0,
\vspace{5pt}\\
{T_2} = 2\beta cq,
\vspace{5pt}\\
{T_3} = ab\beta  + cdq,
\vspace{5pt}\\
{T_4} = 2b\beta q.
\end{array}
\end{equation*}
\vspace{5pt}
Define the basic reproduction number $R_{0}=\dfrac{\lambda \beta}{d a}$. We only take the non-negative solution into account because of the constraints of the real situation. Obviously,
${E_2}(\dfrac{1}{\beta }[a - \dfrac{1}{{2b\beta }}({T_3} - {T_1})], - \dfrac{1}{{{T_2}}}({T_1} + {T_3}), - \dfrac{1}{{{T_4}}}({T_1} + {T_3}))$ is not a non-negative solution, while $E_{1}\left(\dfrac{1}{\beta}\left[a-\dfrac{1}{2 b \beta}\left(T_{3}-T_{1}\right)\right],-\dfrac{1}{T_{2}}\left(T_{3}-T_{1}\right),-\dfrac{1}{T_{4}}\left(T_{3}-T_{1}\right)\right)$  is a positive solution when $a d<\lambda \beta$  i.e. $R_{0}>1$ .

We first analyze equilibrium $E_0$ .\\
\begin{theorem}\label{TH3}
	When $R_0<1$ , $E_0$  is locally asymptotically stable.
\end{theorem}
\begin{proof}
	The Jacobian matrix of system (\ref{E2}) at $E_{0}\left(\dfrac{\lambda}{d}, 0,0\right)$ is
	\begin{equation*}\label{Jocabian_0}
	J\left(E_{0}\right)=\left( \begin{array}{ccc}{-d-\sigma} & {-\beta \dfrac{\lambda}{d}} & {0} \\ {0} & {\beta \dfrac{\lambda}{d}-a-\sigma} & {0} \\ {0} & {c e^{-\sigma \tau}} & {-b-\sigma}\end{array}\right).
	\end{equation*}
	The characteristic equation is
	\begin{equation*}\label{FDFD}
	(\sigma+d)(\sigma+b)\left(\sigma-\left(\beta \frac{\lambda}{d}-a\right)\right)=0.
	\end{equation*}
	Then we obtain the characteristic roots 
	\begin{equation*}
	{\sigma _1} =  - d < 0, {\sigma _2} =  - b < 0, {\sigma _{\rm{3}}} = \beta \dfrac{\lambda }{d} - a.
	\end{equation*}
	Therefore when  $R_{0}<1, \sigma_{3}=\beta \dfrac{\lambda}{d}-a<0$, $E_0$ is locally asymptotically stable;   
	when $R_0>1$, $\sigma_{3}=\beta \dfrac{\lambda}{d}-a>0$, $E_0$  is unstable.
\end{proof}
\begin{theorem}
	When $R_0<1$ , $E_0$ is globally asymptotically stable.
\end{theorem}
\begin{proof}
	In theorem \ref{TH3} We studied the local asymptotic stability of $E_0$, now we discuss its global attractivity.
	
	From eq. (\ref{E2}) we know, when $(y(t), z(t)) \rightarrow(0,0)$, $x(t) \rightarrow \dfrac{\lambda}{d}$, 
	the non-negativity of the solution and the property of eq. (\ref{E2}) yield that $y(t)$, $z(t)$ satisfy the conditions below
	\begin{equation}\label{E3}
	\left\{\begin{array}{c}{D^{\alpha} y(t) \leq \beta \dfrac{\lambda}{d} y-a y},
	\vspace{5pt}\\
	{D^{\alpha} z(t)=c y(t-\tau)-b z}.\end{array}\right.
	\end{equation}
	The right hand side of the first inequality is monotonically increasing with respect to $y$, the right hand side of the second equality is monotonically increasing with respect to $y(t-\tau)$. For system after changing the inequality to equality, if $(y(t), z(t)) \rightarrow(0,0)$ as $t \rightarrow \infty$, then eq. (\ref{E3}) possesses the same property. The characteristic equation of the equality form of system (\ref{E3}) is
	\begin{equation*}\label{fdfdfd}
	(\sigma+b)\left(\sigma-\left(\beta \frac{\lambda}{d}-a\right)\right)=0.
	\end{equation*}
	Consequently the characteristic roots are $\sigma_{1}=-b<0$, $\sigma_{2}=\beta \dfrac{\lambda}{d}-a$.
	
	When $R_0<1$ ,$\sigma_{2}=\beta \dfrac{\lambda}{d}-a<0$, $E_0$ has global attractivity, hence it holds global asymptotic stability.
	
	When $R_0>1$ ,$\sigma_{2}=\beta \dfrac{\lambda}{d}-a>0$, $E_0$ does not possess global attractivity.
\end{proof}

Now we analyze the infection equilibrium $E_1$.

Let $E_{1}=\left(X_{1}, Y_{1}, Z_{1}\right)$ then the Jocabian matrix at $E_1$ is
\begin{equation*}\label{Jocabian_1}
  J\left(E_{1}\right)=\left( \begin{array}{ccc}{-d-\beta Y_{1}-\sigma} & {-\beta X_{1}} & {0} \\ {\beta Y_{1}} & {\beta X_{1}-a-q Z_{1}-\sigma} & {-q Y_{1}} \\ {0} & {c e^{-\sigma \tau}} & {-b-\sigma}\end{array}\right).
\end{equation*}
The characteristic equation is
\begin{equation*}\label{KGBFV}
  {\sigma ^3} - (A - B){\sigma ^2} - (C - D - W{e^{ - \sigma \tau }})\sigma  + E - F + G{e^{ - \sigma \tau }} = 0,
\end{equation*}
where
\begin{equation*}\label{LKJHGF}
\begin{array}{l}
A = {X_1}\beta,
\vspace{5pt}\\
B = b + d + a + \beta {Y_1} + q{Z_1},
\vspace{5pt}\\
C = {X_1}b\beta  + {X_1}\beta d,
\vspace{5pt}\\
D = ad + bd + ab + {Y_1}a\beta  + {Y_1}b\beta  + {Y_1}cp + {Z_1}bp + {Z_1}dp + {Y_1}{Z_1}\beta p,
\vspace{5pt}\\
W = {Y_1}pc,
\vspace{5pt}\\
E = {Y_{\rm{1}}}ab\beta  + abd + {Z_1}bdp + {Y_1}{Z_1}b\beta p,
\vspace{5pt}\\
F = {X_{\rm{1}}}b\beta d,
\vspace{5pt}\\
G = Y_1^2\beta p + {Y_1}dp.
\end{array}
\vspace{5pt}
\end{equation*}
The three characteristic roots are associated with the time delay parameter $\tau$ but are independent of the fractional-order $\alpha$.
 According to the Hurwitz stability criterion \cite{routh_AHMED20061,Hurwitz1895}, when $\alpha$ is small enough, the eigenvalue $\lambda$ of the equation is more likely to satisfy the condition $|\arg (\sigma)|>\dfrac{\alpha \pi}{2}$, then  $E_1$ has local asymptotical stability. Therefore,\\
i) if $E - F + G = 0$, then there exists zero solutions to the equation, and $E_1$ is unstable by theorem \ref{TH1}; \\
ii)  if $E - F + G \ne 0$, then there is no zero solution to the equation. $E_1$ is locally asymptotically stable when  $\alpha$ is small enough.
\section{Numerical simulation}
For system (\ref{E2}), conforming to reference \cite{WangKaiFaMain}, we take the parameters of the integer-order system in chaotic state $\beta=0.002, d=0.1, q=0.05, a=5, c=0.2, e=0.3, \lambda=1500$ and apply the Adames-Bashforth-Moulton method \cite{Diethelm2002} to numerically solve the Caputo fractional differential equations.
\subsection{Effect of time-delay on the fractional HBV model}
Tentatively we set the fractional-order to be 0.985 and investigate time-delay intervals $\tau \in[5,11]$ and $\tau \in[11,17]$ respectively. Fig. \ref{F1} shows bifurcation diagrams and the largest Lyapunov exponent versus the immune time delay. As presented in fig.\ref{F1}, the alteration of the HBV population $y$  suggests that the Caputo fractional HBV system has abundant dynamical properties. From fig.\ref{F1} (a), it is clear that system (\ref{E2}) reaches a chaotic state through period-doubling bifurcation, displaying distinct chaotic behavior in the interval $\tau \in[9,10.1]$. From fig. \ref{F1}(b), the system returns to a stable state after entering the chaos for the first time, and then enters chaotic state again during time lag $\tau \in[15.2,15.6]$. The change of the virus load of fig. \ref{F1} illustrates the course of the infection of healthy liver cells. The corrsponding clinical situation is the patient's infection with hepatitis B virus, the subsequent incubation period, and HBV's rapid proliferation after the first infection. It can be seen that time delay brings about more complicated chaotic characteristics. And these numerical results are clinically consistent with the common situations of the recurrence of HBV after healed \cite{CALVARUSO2017898,XiePrevention}. Compared with the integer-order system \cite{WangKaiFaMain}, the fractional-order system is closer to the real evolution of hepatitis B in vivo.
\begin{figure}[htbp]   
  \centering
  \subfigure[]{
    \label{fig:subfig:onefunction} 
    \includegraphics[scale=0.5]{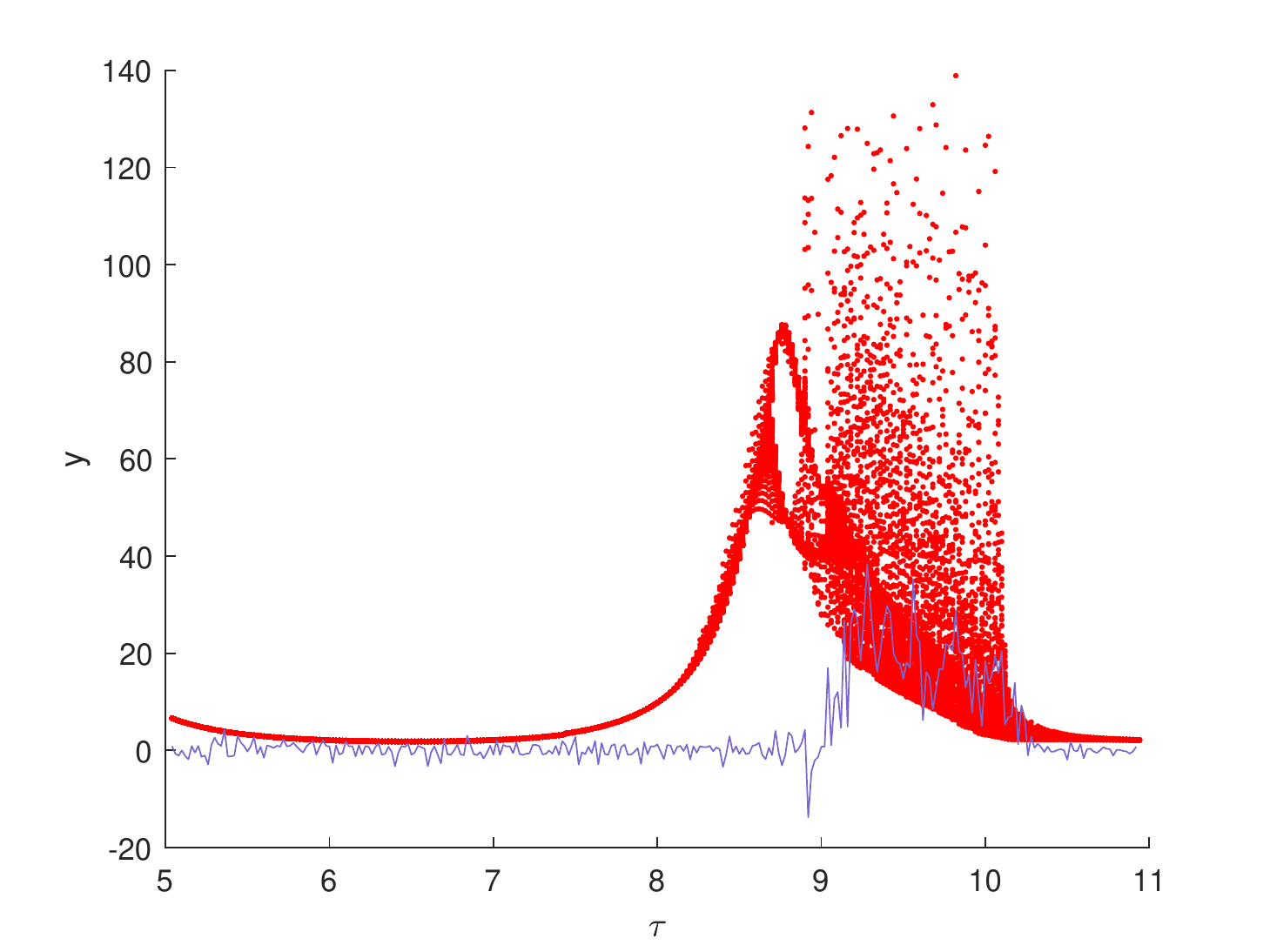}}
  \hspace{0in}
  \subfigure[]{
    \includegraphics[scale=0.5]{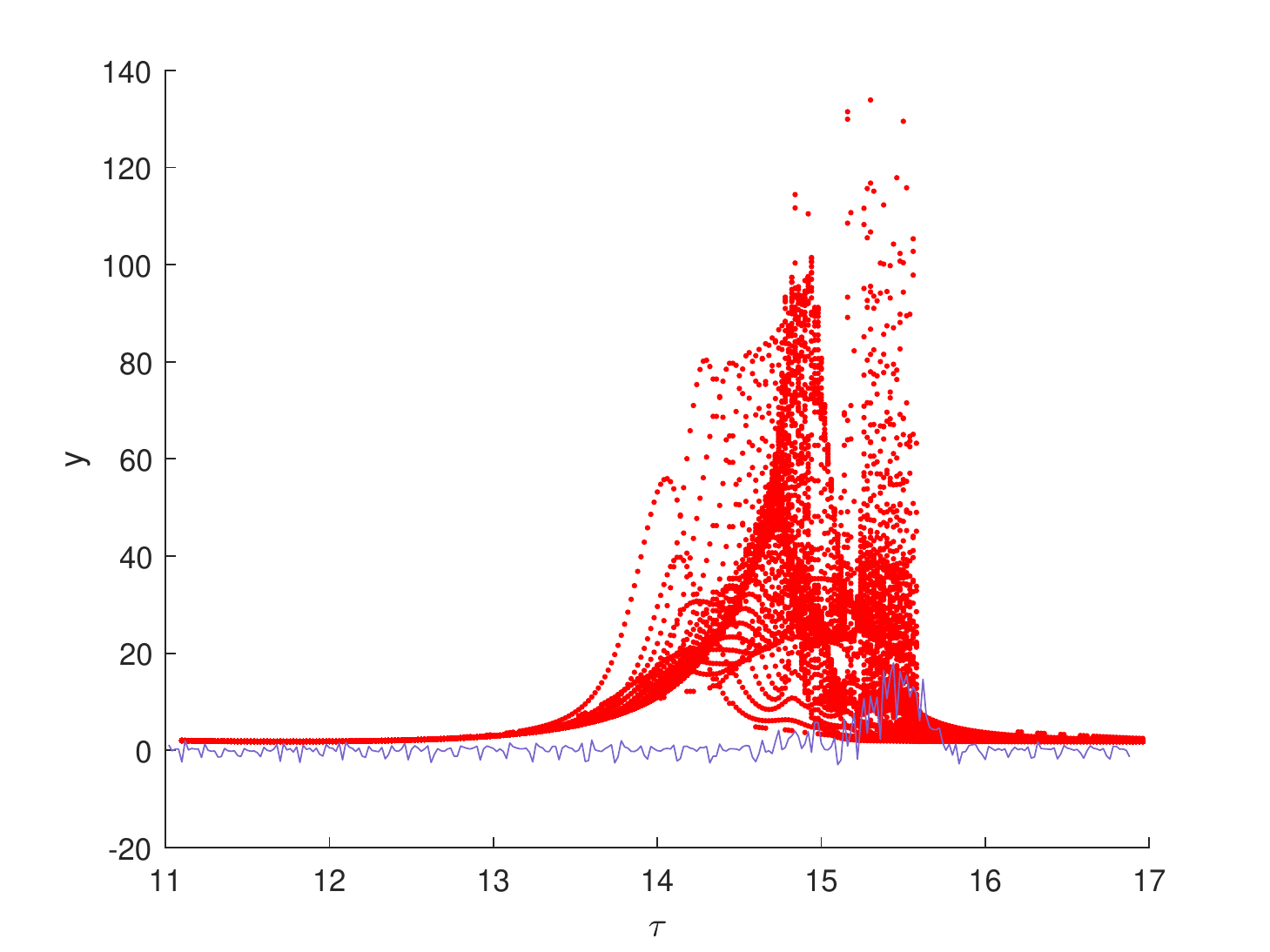}}
  \caption{Bifurcation diagrams and largest Lyapunov exponent diagrams of y with time delay:\\
(a) $\tau \in[5,11]$ ; (b) $\tau \in[11,17]$.
}
  \label{F1} 
\end{figure}
\subsection{Effect of fractional-order on the fractional HBV model}
Because fractional HBV system producing chaos represents diffusion of hepatitis B virus and deterioration of the disease, the development of Hepatitis B should be controlled as early as possible. Effective treatment measures should be taken during the medical process of the first onset of hepatitis B. For this reason, this paper focuses on the critical period of the hepatitis B immune response, i.e. the first time  the system is in a chaotic state. 
Taking time-delay parameter $\tau = 9 $ with a noticeable chaotic characteristics, we observe the influence of the order on the fractional system. Fig. \ref{F2}(a) and (b) are the bifurcation diagrams and the largest Lyapunov exponent diagrams versus $x$ and $y$, respectively. We can learn from fig. \ref{F2}(a) that when the fractional-order $\alpha \in[0.951,0.986]$, the fractional HBV system presents conspicuous chaotic characteristics. And as depicted in fig. \ref{F2}(b), when the fractional-order $\alpha \in[0.953,0.986]$, the image also shows obvious chaotic behavior. In these two intervals, we can also discover dark lines appear within the ranges of the variables. The change of largest Lyapunov exponent infers that in the neighbourhood of bifurcation points the state of system becomes extremely sensitive to fractional-order changes. Thereby, it makes great sense to study the influence of the order on the fractional HBV system. In this paper we take fractional-order parameter $\alpha = 0.985$ to further explore the fractional HBV model.\\
\begin{figure}[htbp]
  \centering         
  \subfigure[]{
    \label{fig:subfig:onefunction} 
    \includegraphics[scale=0.5]{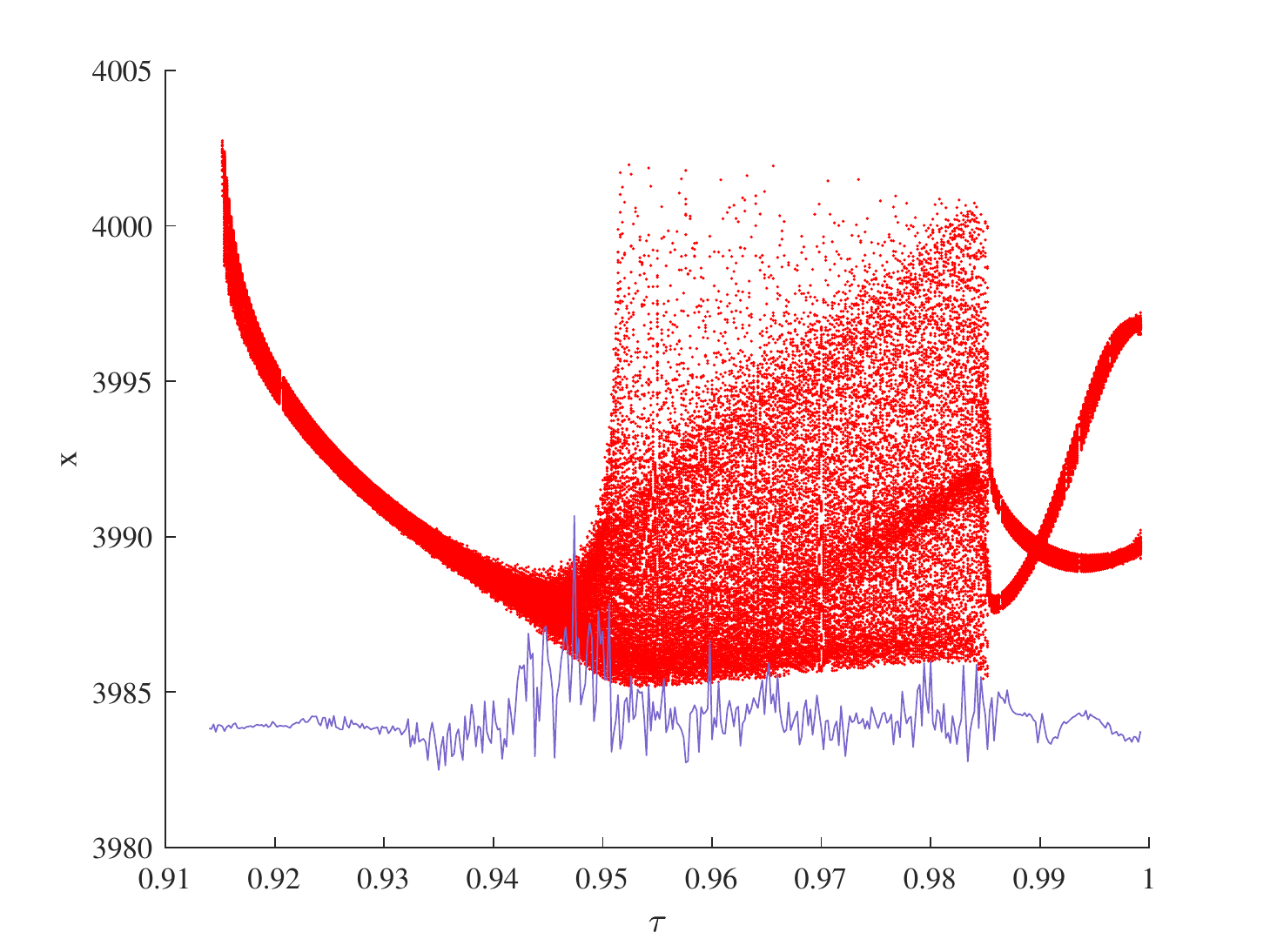}}
  \hspace{0in}
  \subfigure[]{
    \includegraphics[scale=0.5]{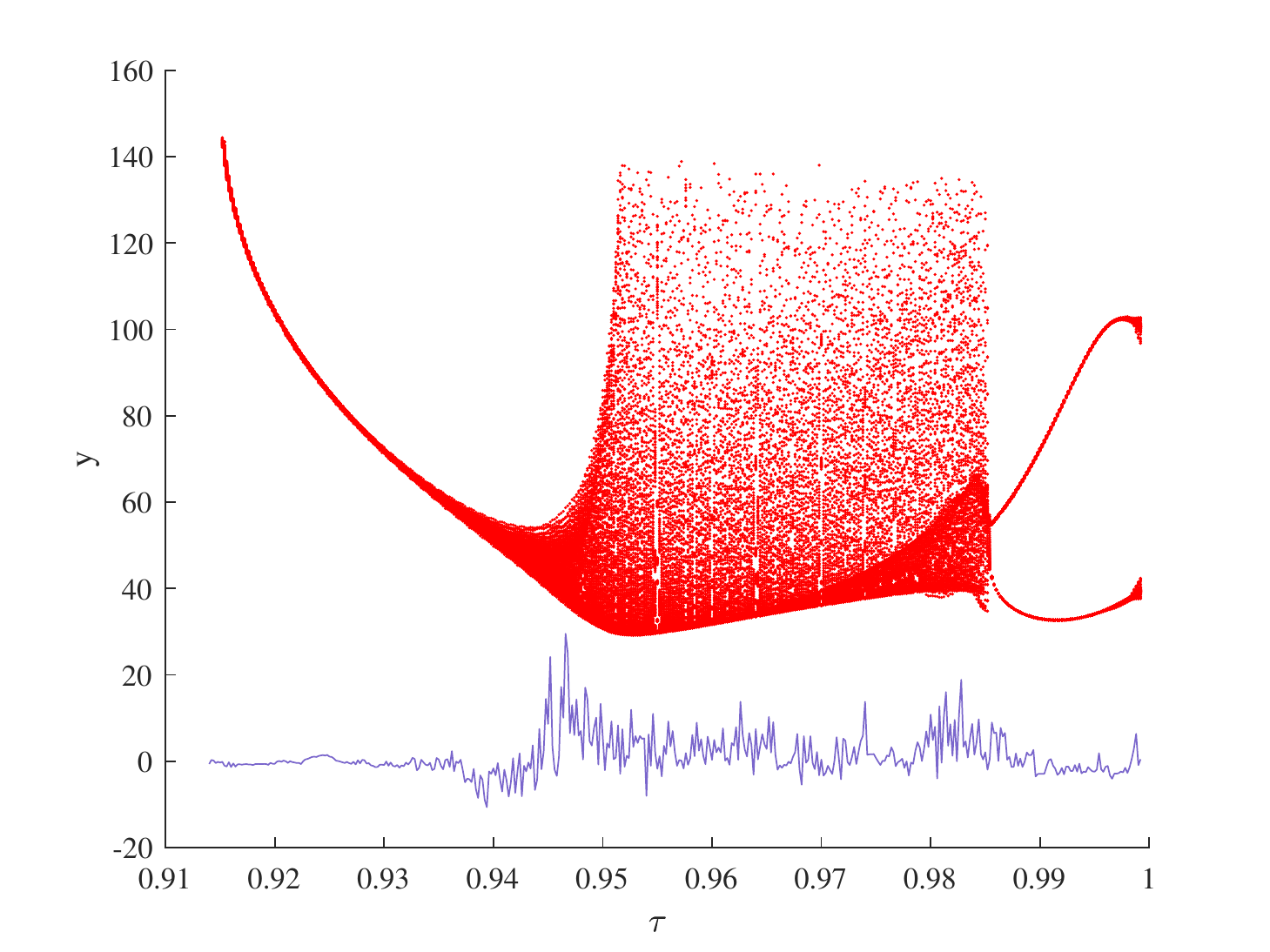}}
  \caption{Bifurcation diagrams and largest Lyapunov exponent diagrams of system (\ref{E2}) with fractional-order $\alpha \in[0.91,1]$: (a) bifurcation diagram and largest Lyapunov exponent diagram of x; (b) bifurcation diagram and largest Lyapunov exponent diagram of y.}
  \label{F2} 
\end{figure}
\subsection{Chaotic characteristics}
In order to investigate the specific effect of time delay on the stability of the system, two time-delay parameters $\tau=7$ and $\tau=9$ are selected respectively while other parameter values remain unchanged. Fig. \ref{F3} shows the phase diagram and Poincar\'e section of system (\ref{E2}) at order $\alpha=0.985$, time delay $\tau=7$.  
\begin{figure}[htbp]
  \centering
  \subfigure[]{
    \label{fig:subfig:onefunction} 
    \includegraphics[scale=0.5]{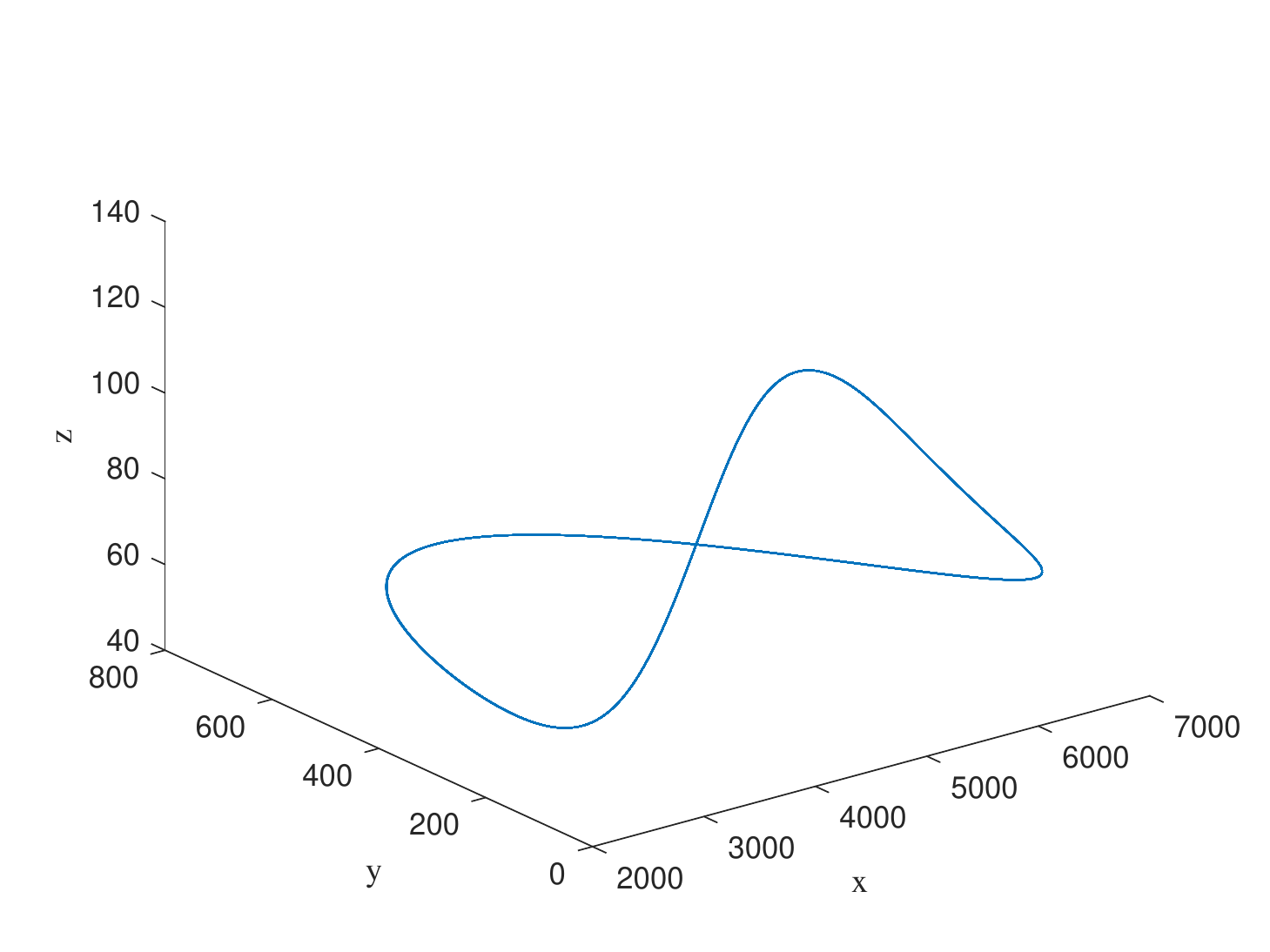}}
  \hspace{0in}
  \subfigure[]{
    \includegraphics[scale=0.5]{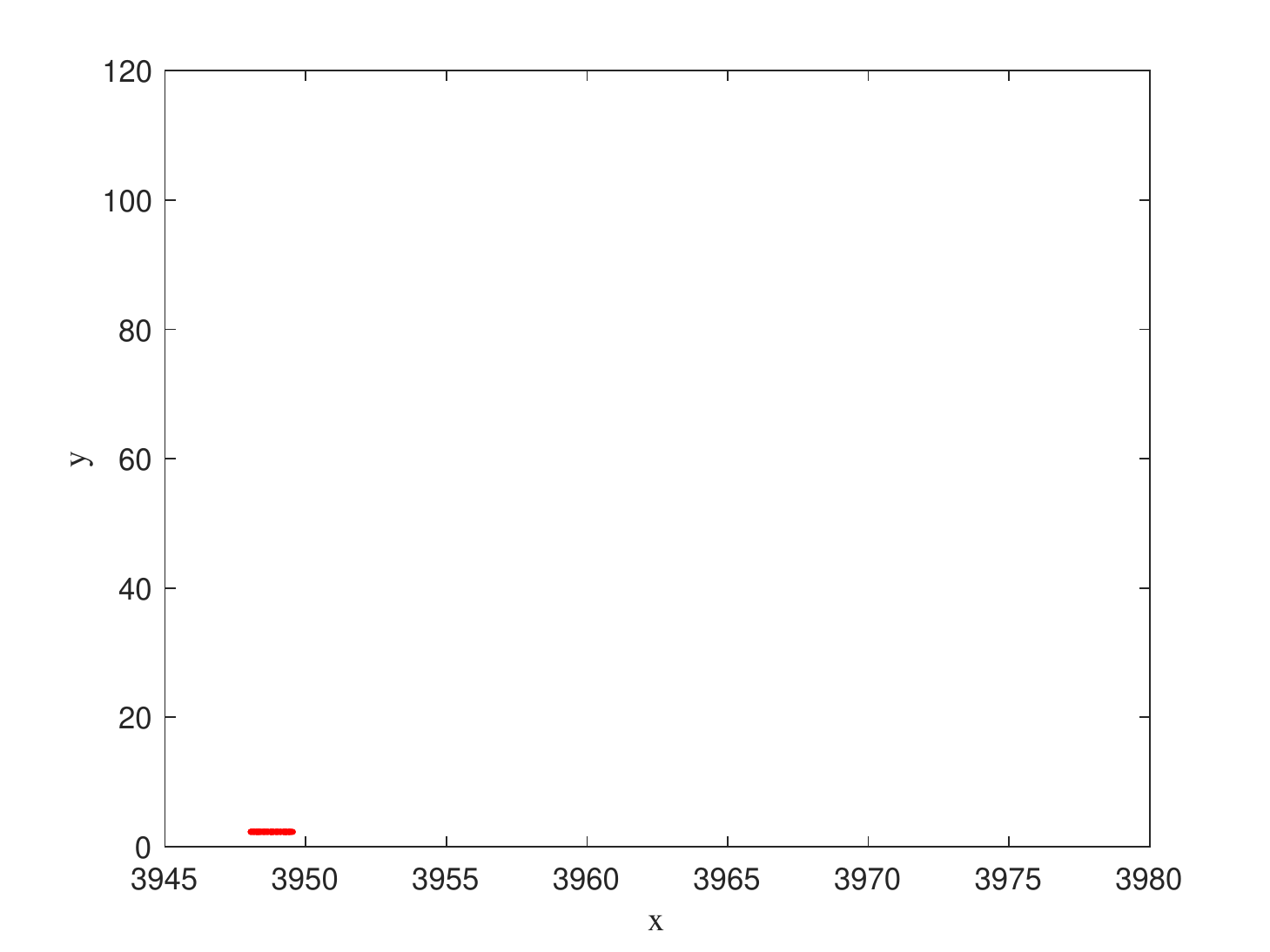}}
  \caption{System order in the case $\alpha=0.985$ ; (a) phase diagram; (b) Poincar\'e section.}
  \label{F3} 
\end{figure}
Poincar\'e section possessing only one fixed point or a few discrete points indicates the system is periodic, and posscessing a closed curve indicates the system is in a quasi-periodic state. Thus it can be observed that the system is quasi-periodic in fig. \ref{F3}. Additionally, if the Poincar\'e section possess dense points with hierarchical structure, the system is in a chaotic state. From fig. \ref{F4}, it can be judged that when the immune delay  $\tau=9$, the system is chaotic.
\begin{figure}[H]
  \centering
  \subfigure[]{
    \label{fig:subfig:onefunction} 
    \includegraphics[scale=0.6]{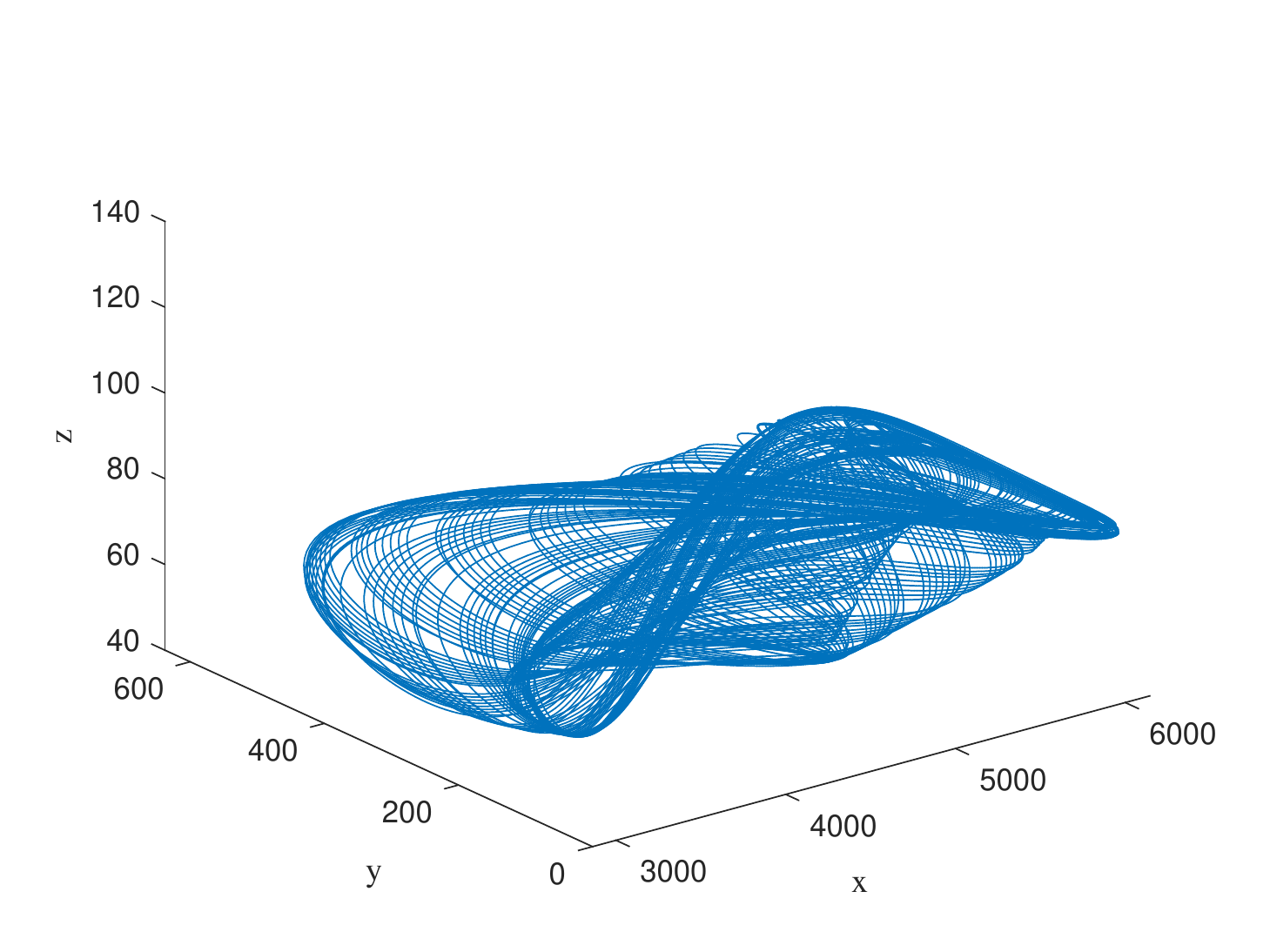}}
  \hspace{0in}
  \subfigure[]{
    \includegraphics[scale=0.5]{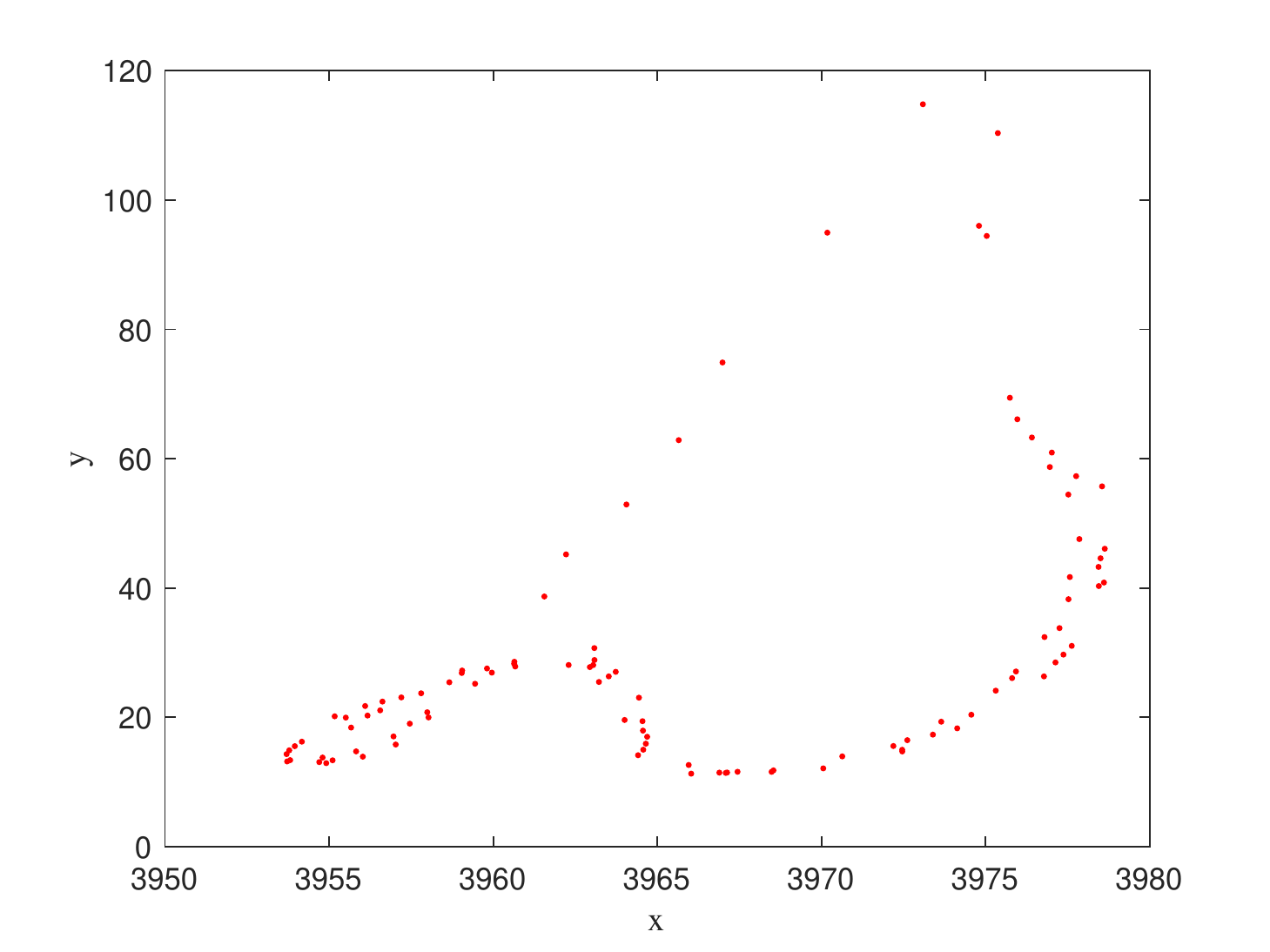}}
  \caption{System order in the case $\alpha=0.985$ ; (a) phase diagram; (b) Poincar\'e section.}
  \label{F4} 
\end{figure}
\begin{figure}[h]
  \centering
  \subfigure[]{
    \label{fig:subfig:onefunction} 
    \includegraphics[scale=0.5]{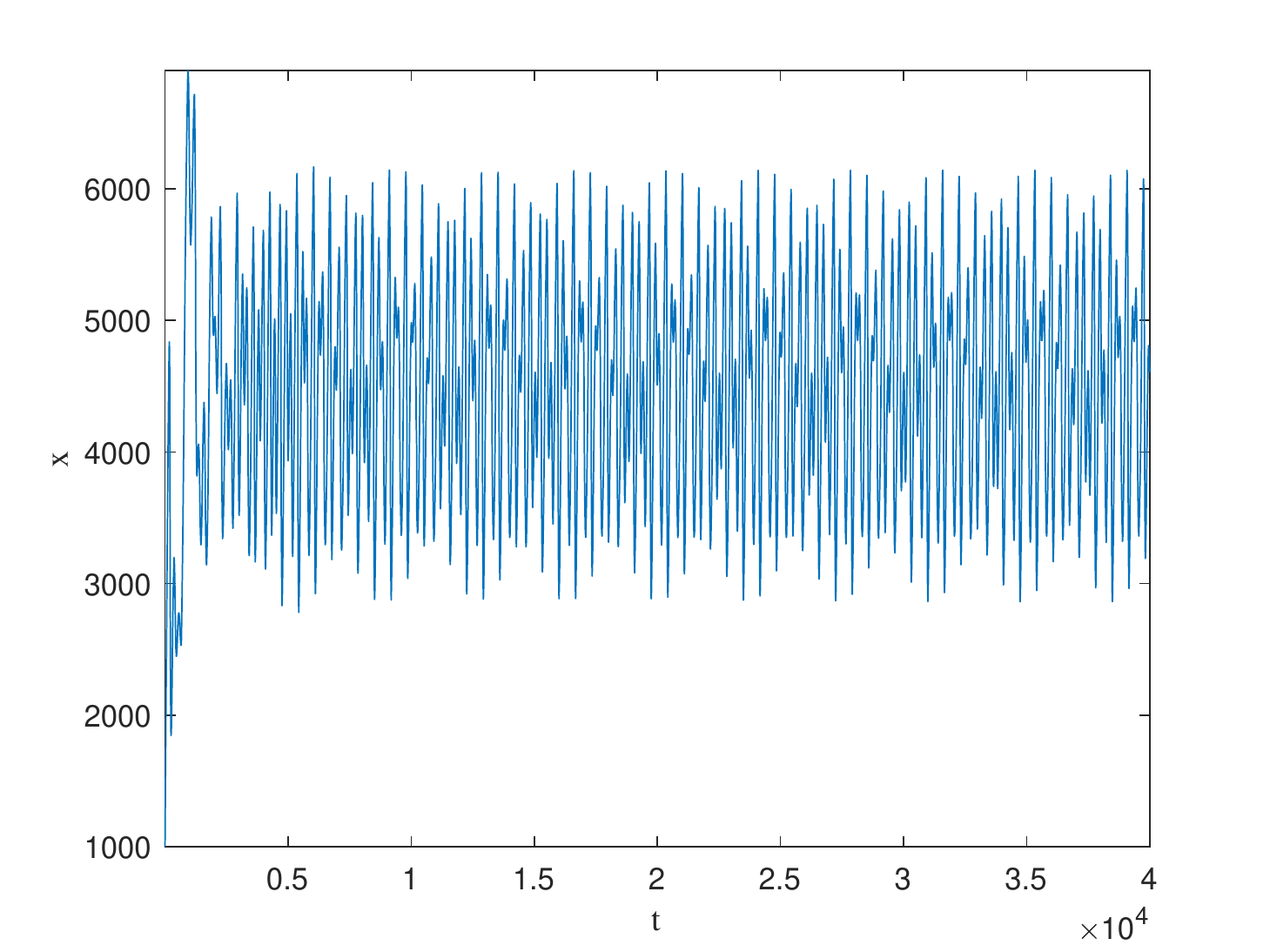}}
  \subfigure[]{
    \label{fig:subfig:onefunction} 
    \includegraphics[scale=0.5]{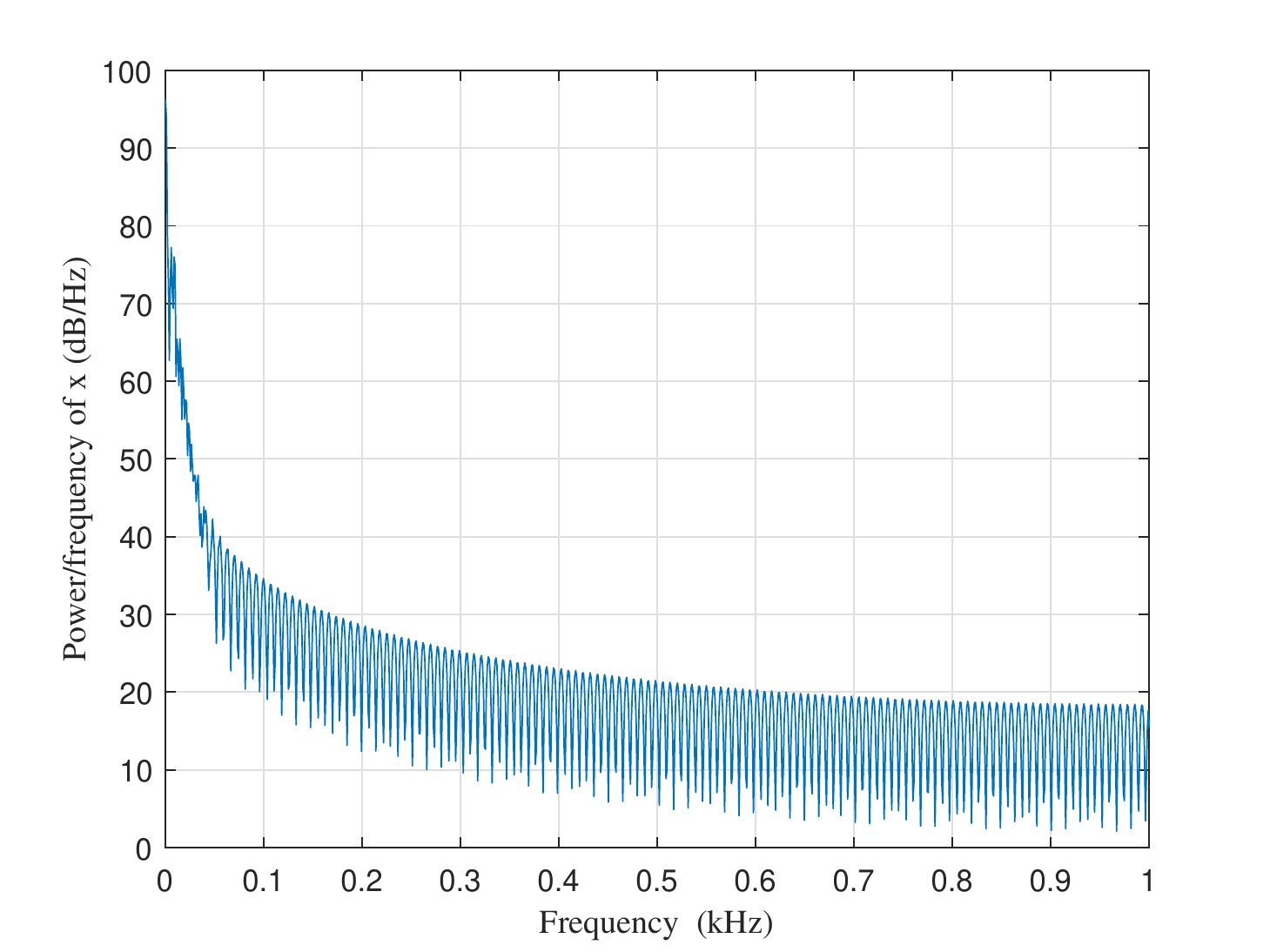}}
  \subfigure[]{
    \label{fig:subfig:onefunction} 
    \includegraphics[scale=0.5]{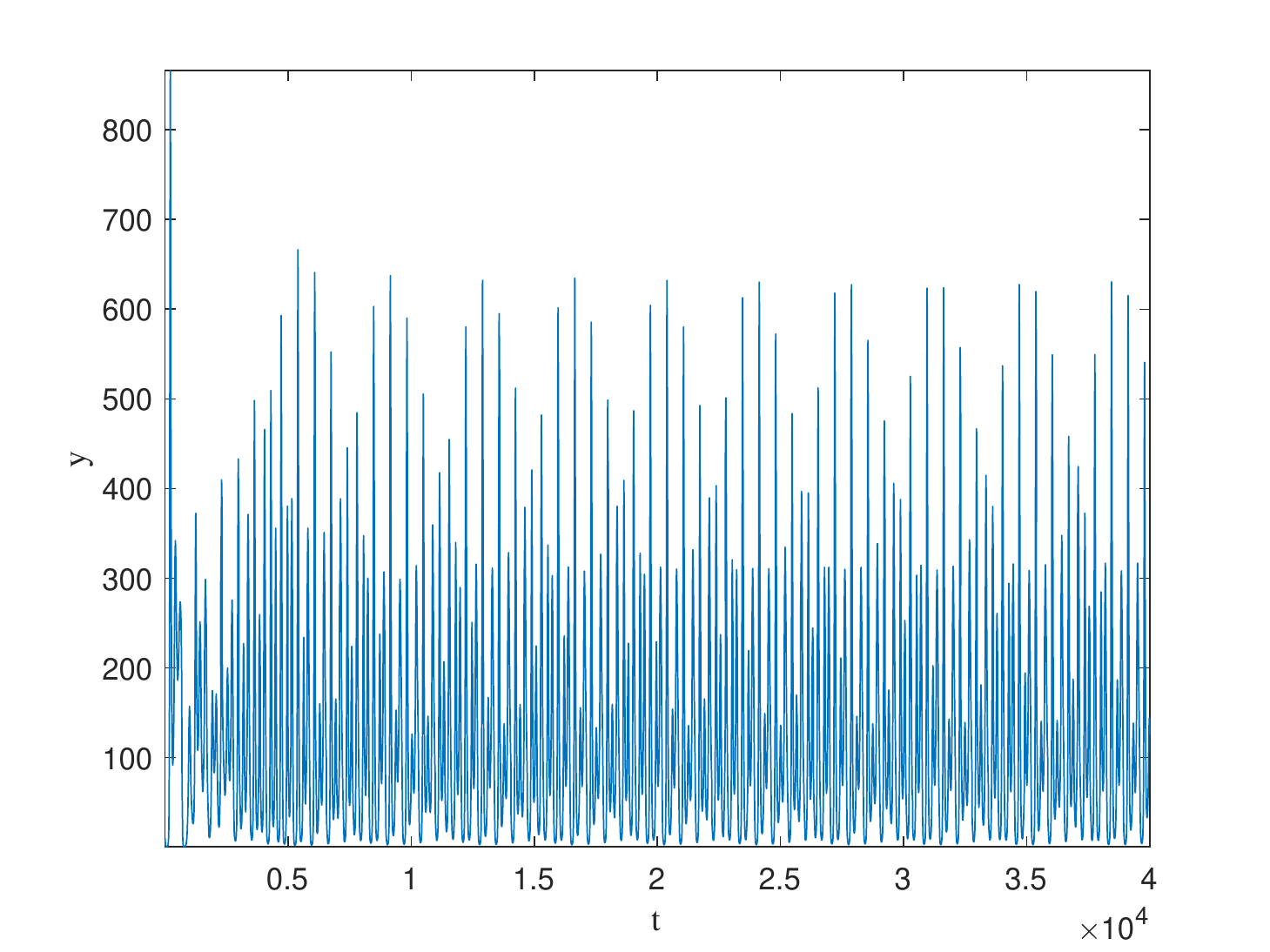}}
  \subfigure[]{
    \label{fig:subfig:onefunction} 
    \includegraphics[scale=0.5]{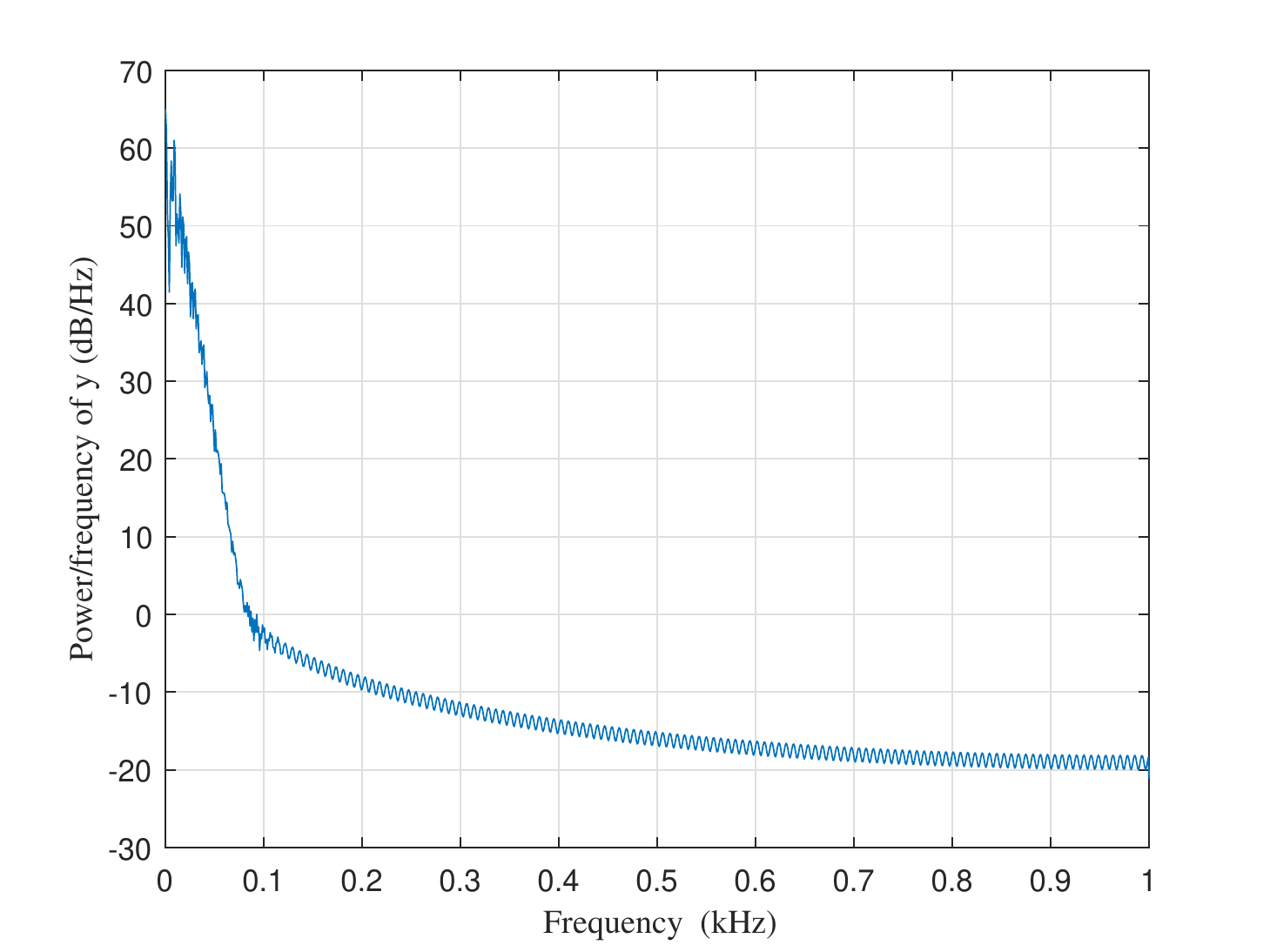}}
  \subfigure[]{
    \label{fig:subfig:onefunction} 
    \includegraphics[scale=0.5]{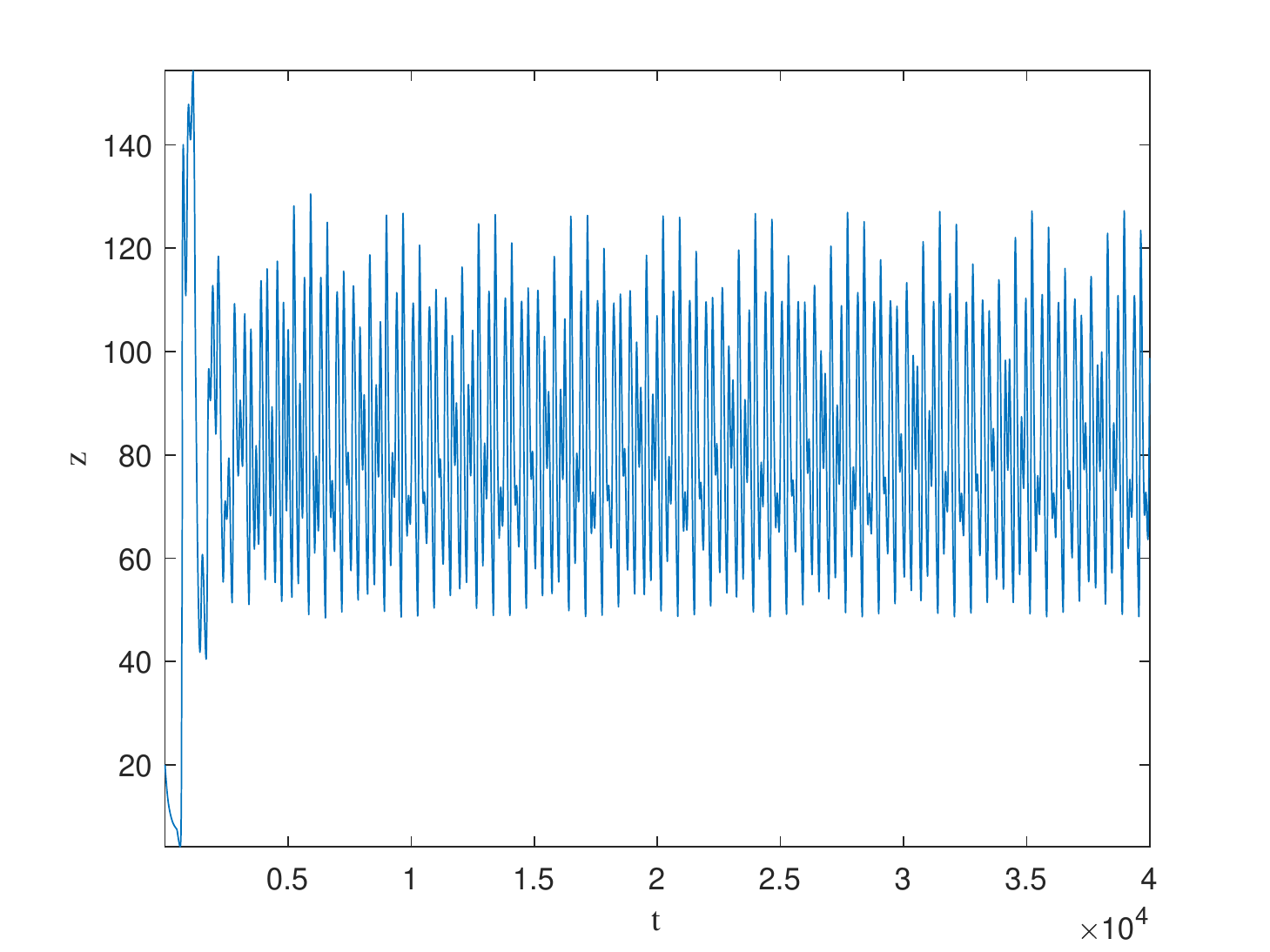}}
  \subfigure[]{
    \label{fig:subfig:onefunction} 
    \includegraphics[scale=0.5]{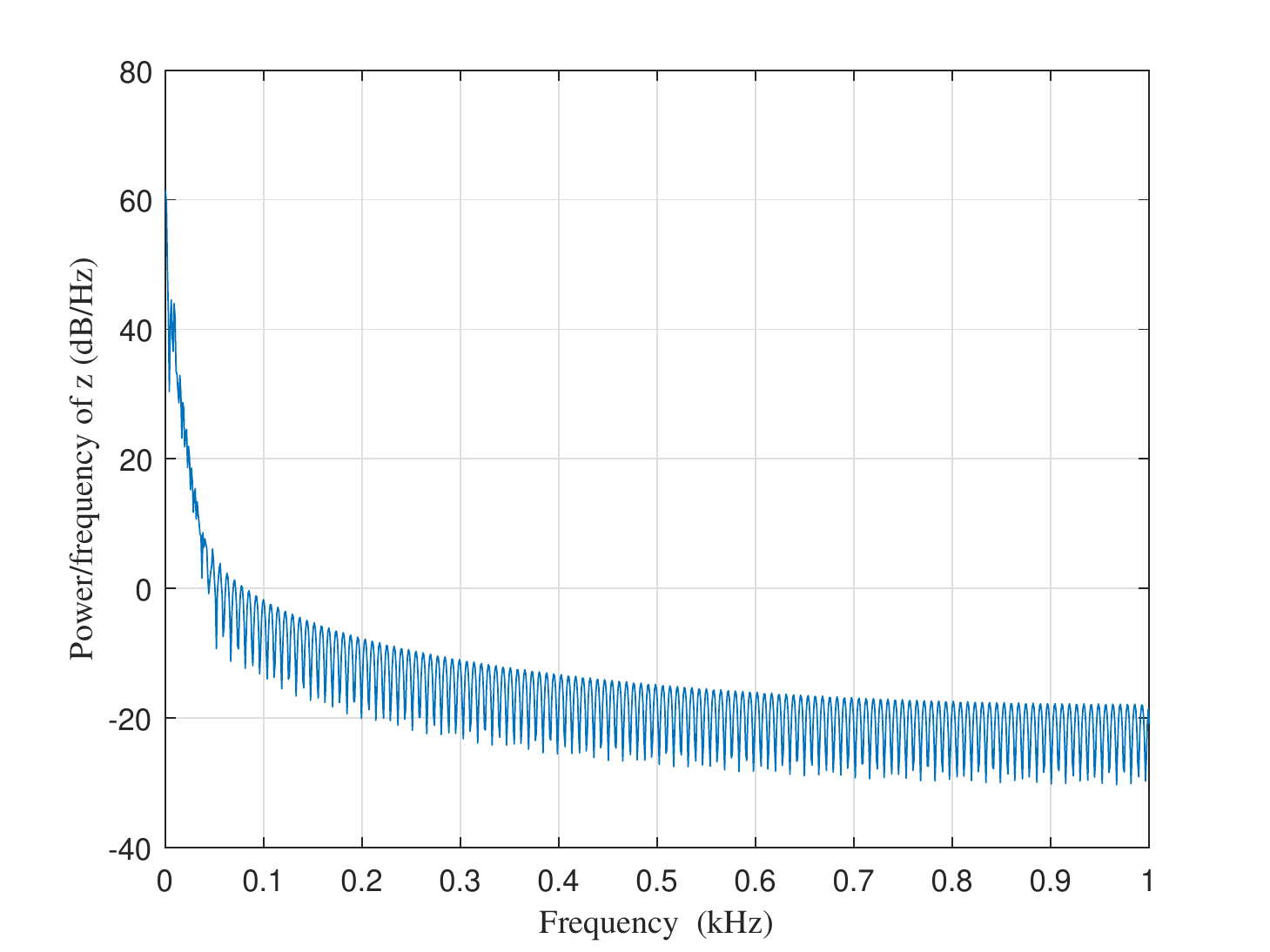}}
  \caption{Simulation of fractional HBV model with time delay $\tau=9.5$ (a) Curve of $x-t$; (b) power spectrum; (c) Curve of $y-t$; (d) power spectrum; (e) Curve of $z-t$; (f) power spectrum.}
  \label{F5} 
\end{figure}
To further expound the dynamics of the Caputo fractional HBV system, we characterize system (\ref{E2}) by plotting time courses and power spectrums. Fig. \ref{F5} displays curves of $x-t$(fig. \ref{F5}(a)), $y-t$(fig. \ref{F5}(c)), $z-t$(fig. \ref{F5}(e)), and power spectrum diagrams of $x,y,z$. As a result, phase diagram in fig. \ref{F5}(a), dense points with hierarchies on the Poincar\'e section in fig. \ref{F5}(b), and continuous spectra with broad peak of noise background shown in the power spectra of Fig. \ref{F5} all indicates that the system appears complex dynamics of chaotic phenomenon when time-delay $\tau=9$. It is found that under unstable physiological conditions, consistent with integer-order HBV systems, fractional-order systems can also be chaotic.
\section{Asymptotical stability theory of Caputo fractional system with delay}
Wu \cite{WuXiang} introduced a criterion on asymptotic stability of Riemann-Liouville fractional differential nonlinear systems, but it is not clear whether the theorem is applicable to Caputo fractional differential. In order to explore the theoretical nature of Caputo fractional-order systems with time delay, we extend the theoretical results to the case of Caputo fractional differential. We give a new theory start with the following lemmas.
\begin{lemma}\label{lemma3}
	\cite{Liu2015} The inequality
	\begin{equation*}\label{LKJHGF}
	2{x^T}y \le \varepsilon {x^T}x + \frac{1}{\varepsilon }{y^T}y
	\end{equation*}
	holds for any
	\begin{equation*}
	x, y \in \mathbb{R}^{n}, \varepsilon>0.
	\end{equation*}
\end{lemma}
\begin{lemma}\label{lemma4}
	\cite{TAN20081010} For real symmetric matrixes $U>0, V \geq 0$ , there exists
	\begin{equation*}\label{dfsdfds}
	\begin{array}{l}
	U > V \Leftrightarrow {\lambda _{\max }}(V{U^{ - 1}}) < 1 \Leftrightarrow {\lambda _{\max }}({U^{ - \frac{1}{2}}}V{U^{ - \frac{1}{2}}}) < 1.
	\end{array}
	\end{equation*}
\end{lemma}
\begin{theorem}
	Suppose fractional-order nonlinear system
	\begin{equation}\label{E4}
	{}_{{t_0}}D_t^\alpha X(t) = AX(t) + BX(t - \tau ) + F(X(t)),
	\end{equation}
	where $0<\alpha<1$, $x(t) \in {\mathbb{R}^n}$ is the state vector of the system; $A,B \in {\mathbb{R}^{n \times n}}$ are constant matrixes; time delay $\tau \geq 0$.
	The nonlinear part satisfies $F(0)=0$, and $F(\cdot)$ is the higher-order term of $(\cdot)$, 
	i.e. $\mathop {\lim }\limits_{\left\| X \right\| \to 0} \dfrac{{\left\| {F(X)} \right\|}}{{\left\| X \right\|}} = 0$.
	
	If there exits two symmetric and positive definite matrixes $P,Q$ satisfying the equality and inequality below
	\begin{equation}\label{E5}
	PA + {A^T}P + 2Q = 0,{\kern 1pt} {\kern 1pt} {\kern 1pt} {\kern 1pt} {\kern 1pt} {\kern 1pt} {\kern 1pt} {\kern 1pt}
	\end{equation}
	\begin{equation}\label{E6}
	\left\| {PB} \right\| < {\lambda _{\min }}(Q),{\kern 1pt} {\kern 1pt} {\kern 1pt} {\kern 1pt} {\kern 1pt} {\kern 1pt} {\kern 1pt} {\kern 1pt}
	\end{equation}
	then the fractional system (\ref{E4}) is asymptotically stable.
\end{theorem}
\begin{proof}
		We construct Lyapunov function \\
	\begin{equation*}\label{51}
	V(t) = {}_{{t_0}}^CD_t^\alpha ({x^T}(t)Px(t)) + \int_{t - \tau }^t {{x^T}(s)Qx(s)} ds.
	\end{equation*}
	By lemma \ref{lemma2}, we obtain that the derivative of functional $V(t)$ along the trajectory of system (\ref{E4}) satisfies
	\begin{align}\label{52}
	\dot{V}(t) &= {}_{{t_0}}^CD_t^1\left\{ {{}_{{t_0}}^CD_t^{\alpha  - 1}({x^T}(t)Px(t)) + \int_{t - \tau }^t {{x^T}(s)Qx(s)} ds} \right\}
	\vspace{5pt}\nonumber\\
	& = {}_{{t_0}}^CD_t^\alpha ({x^T}(t)Px(t)) + {x^T}(t)Qx(t) - {x^T}(t - \tau )Qx(t - \tau )
	\vspace{5pt}\nonumber\\
	& \le 2{x^T}(t)P{}_{{t_0}}^CD_t^\alpha {x^T}(t) + {x^T}(t)Qx(t) - {x^T}(t - \tau )Qx(t - \tau )
	\vspace{5pt}\nonumber\\
	& = {x^T}(t)(PA + {A^T}P + Q)x(t) + 2{x^T}(t)PBx(t - \tau ) + 2{x^T}(t)PF(x(t)) - {x^T}(t - \tau )Qx(t - \tau ).
	\end{align}
	From lemma \ref{lemma3}, we have
	\begin{equation*}\label{53}
	\begin{split}
	2{x^T}(t)PBx(t - \tau ) &= 2{x^T}(t)PB{Q^{ - \frac{1}{2}}}{Q^{\frac{1}{2}}}x(t - \tau )
	\vspace{5pt}\notag\\
	&\le \frac{1}{\alpha }{x^T}(t)PB{Q^{ - 1}}{B^T}Px(t) + \alpha {x^T}(t - \tau )Qx(t - \tau ),\\
	2{x^T}(t)PF(x(t)) &\le \frac{1}{\beta }{x^T}(t){P^2}x(t) + \beta F_{}^T(x(t))F(x(t)),\notag
	\end{split}
	\end{equation*}
	where $\alpha , \beta$ are arbitrary positive constants. Thus, by eq. (\ref{E5}) and eq. (\ref{52}) there is
	\begin{equation*}\label{55}
	\begin{array}{l}
	\dot{V}(t)\le {x^T}(t)(\dfrac{1}{\alpha }{x^T}(t)PB{Q^{ - 1}}{B^T}P - Q + (\dfrac{1}{\beta }){P^2})x(t)
	+ \beta F_{}^T(x(t))F(x(t)) + (\alpha  - 1){x^T}(t - \tau )Qx(t - \tau ).
	\end{array}
	\end{equation*}
	On account of the definition of spectral norm of matrix, we attain
	\begin{equation*}\label{56}
	\begin{array}{l}
	{[{\lambda _{\max }}({Q^{ - \frac{1}{2}}}PB{Q^{ - 1}}{B^T}P{Q^{ - \frac{1}{2}}})]^{\frac{1}{2}}}
	\vspace{5pt}\\
	= \left\| {{Q^{ - \frac{1}{2}}}PB{Q^{ - \frac{1}{2}}}} \right\|
	\vspace{5pt}\\
	\le {\left\| {{Q^{ - \frac{1}{2}}}} \right\|^2}\left\| {PB} \right\|
	\vspace{5pt}\\
	\le \dfrac{1}{{{\lambda _{\min }}(Q)}}\left\| {PB} \right\|.
	\end{array}
	\end{equation*}
	By inequality (\ref{E6}), we get
	\begin{equation*}\label{57}
	{\lambda _{\max }}({Q^{ - \frac{1}{2}}}PB{Q^{ - 1}}{B^T}P{Q^{ - \frac{1}{2}}}) < 1
	\vspace{5pt}.
	\end{equation*}
	Therefore there exists $\eta>0$ , such that ${\lambda _{\max }}({Q^{ - \frac{1}{2}}}PB{Q^{ - 1}}{B^T}P{Q^{ - \frac{1}{2}}}) < \eta  < 1.$
	Because $P > 0,B{Q^{ - 1}}{B^T} \ge 0$ , from lemma \ref{lemma4}, we have $PB{Q^{ - 1}}{B^T}P < \eta Q$.
	For any $\alpha<1$ that meets $0 < \frac{\eta }{\alpha } < 1$, there is $\frac{1}{\alpha }PB{Q^{ - 1}}{B^T}P - Q < (\frac{\eta }{\alpha } - 1)Q < 0.$
	We can select suitable positive constant $\beta $  which satisfies
	\begin{equation}\label{58}
	{M_1}: = \frac{1}{\alpha }PB{Q^{ - 1}}{B^T}P - Q + \frac{1}{\beta }{P^2} < 0.
	\end{equation}
	Note that $\alpha<1 $, whereupon
	\begin{equation}\label{59}
	{M_2}{\rm{: = }}(\alpha  - 1)Q < 0,
	\end{equation}
	consequently
	\begin{equation}\label{510}
	\dot{V}(t) \le {x^T}(t){M_1}x(t) + {x^T}(t - \tau ){M_2}x(t - \tau ) + \beta {\left\| {F(x(t))} \right\|^2}.
	\end{equation}
	From eq. (\ref{58}) and eq. (\ref{59}) , we can take $\sigma >0$ so that 
	\begin{equation}\label{511}
	{M_1} + \sigma I < 0,{M_2} + \sigma I < 0.
	\end{equation}
	In that $F(\cdot)$ is the higher-order term of $(\cdot)$, there exists $\delta>0$ satisfies the condition that when $\left\| {x(t)} \right\| < \delta , t \ge {t_0}$, the inequality
	\begin{equation}\label{512}
	{\left\| {F(x(t))} \right\|^2} \le \frac{\sigma }{\beta }{\left\| {x(t)} \right\|^2}
	\end{equation}
	holds.\\
	Plugging (\ref{512}) into (\ref{510}) yields
	\begin{equation}\label{513}
	\dot{V}(t) \le {x^T}(t)({M_1} + \sigma I)x(t) + {x^T}(t - \tau )({M_2} + \sigma I)x(t - \tau ).
	\end{equation}
	\vspace{5pt}
	By (\ref{511}), $\dot{V}$ is negative definite, therefore the fractional-order system (\ref{E4}) is asymptotically stable.
\end{proof} 
\section{Conclusion}
In view of the existing integer-order HBV models' deficiency in depicting system dynamical properties, we proposed a fractional hepatitis B virus model with immune delay based on Caputo fractional derivative. We discussed the stability of disease equilibria, and explored their immune mechanism. Through theoretical analysis and numerical simulation, the rich dynamic behavior of the system under the influence of time delay and fractional-order were demonstrated. By investigating the impact of time delay in the course of immunization, we selected an essential time delay parameter, then applied classical chaos research methods to study the chaotic characteristics of the fractional-order system. Time sequence maps, phase diagrams, Poincar\'e sections, and power spectrums were derived to describe the chaotic state with different parameter values. Also, we utilized bifurcation diagrams for time delay and fractional-order to analyze their effect on the fractional delay system. Furthermore, motivated by the asymptotic stability criteria for Riemann Liouville fractional-order nonlinear systems, we give an asymptotic stability theorem of the Caputo fractional-order nonlinear autonomous system with time delay. Numerical simulation show that in the short term HBV viral load increases with the immune delay while the fractional HBV system undergo period-doubling bifurcation to chaos. Soon after, the system went through a steady state and then evolved into chaotic state again. This result accords with the common clinical situations of HBV recurrence after the first time it is cured. Thereby it provides clinical guidance that medical care personnel should pay close attention to patients who just cured of HBV to prevent relapse. And in order to avoid those who get acute infected develop chronic Hepatitis B, medical care personnel should consolidate the medication to stabilize the condition of the patients. In addition, the patients ought to exercise properly to enhance their physical fitness and immune function, have regular physical checkups to inspect their health condition, and last but not least stay away from tobacco and alcohol and develop healthy habits. For people who are not infected by HBV, an effective means to protect from HBV is to be vaccinated. Besides, it is also necessary to develop good hygiene habits and do not share personal items such as toothbrushes, towels, and razors. Have routine medical examination to learn about the health status of the body, and start treatment immediately once noticed abnormal signals. Our work actually forms a good correspondence with clinical facts. This article provides a significant theoretical basis for the study of the immune mechanism and clinical treatment of hepatitis B.

\begin{acknowledgement}
	To corresponding author Prof. Dr. Fei GAO we sincerely present this acknowledgement in token of gratitude for his instructions and patience. This work is supported by the Fundamental Research Funds for the Central Universities of China, the Self-determined and Innovative Research Funds of WUT, China (Grant No. 2018IB017), the State Key Program of the National Natural Science of China (Grant No. 91324201), the National innovation and entrepreneurship training program for college students (WUT No. 201910497192).
\end{acknowledgement} 
%
%
%

%
%
\bibliographystyle{unsrt}
\bibliography{References}

\end{document}